\documentclass[onecolumn,12pt]{IEEEtran} 

\usepackage{amssymb}
\usepackage{amsthm,color}
\usepackage{latexsym}
\usepackage{graphicx}
\usepackage{fancyhdr}
\usepackage{bbm}
\usepackage[cmex10]{amsmath}
\usepackage{latexsym,amssymb,amsxtra,mathrsfs,times}
\usepackage[noadjust]{cite}

\newtheorem{thm}{Theorem}

\theoremstyle{definition}
\newtheorem{defn}{Definition}

\newtheorem{exam}{Example}

\newtheorem*{xrmk}{Remark}

\newcommand{\Z}{{\mathbb Z}}
\newcommand{\FF}{{\mathbb F}}
\newcommand{\F}{{\mathbb F}_q}
\newcommand{\C}{{\mathcal C}}
\newcommand{\wt}{{\mathrm{wt} }}

\newcommand{\D}{d_{\text{\tiny free}}}

\newcommand{\I}{\mathbf I}
\newcommand{\supp}{\mathrm{Supp}}

\newcommand{\epsi}{\epsilon}

\newcommand{\vepsi}{\vec{\epsi}}

\newcommand{\gone}[1]{g^{(#1)}_{\vec{\epsi}}}

\newcommand{\cone}[1]{\C^{(#1)}_{\vec{\epsi}}}
\newcommand{\ctwo}[2]{\C^{(#1)}_{\vec{\epsi}^{(#2)}}}

\newcommand{\legn}[2]{\genfrac{(}{)}{}{}{#1}{#2}}

\renewcommand{\vec}{\underline}

\usepackage{setspace}
\begin{document}




\title{On cyclic codes of composite length and the minimal distance}

\author{Maosheng Xiong
\thanks{M. Xiong is with the Department of Mathematics, Hong Kong University of Science and Technology, Clear Water Bay, Kowloon, Hong Kong (e-mail: mamsxiong@ust.hk).}}

\maketitle


\renewcommand{\thefootnote}{}





\begin{abstract}

In an interesting paper Professor Cunsheng Ding provided three constructions of cyclic codes of length being a product of two primes. Numerical data shows that many codes from these constructions are best cyclic codes of the same length and dimension over the same finite field. However, not much is known about these codes. In this paper we explain some of the mysteries of the numerical data by developing a general method on cyclic codes of composite length and on estimating the minimal distance. Inspired by the new method, we also provide a general construction of cyclic codes of composite length. Numerical data shows that it produces many best cyclic codes as well. Finally, we point out how these cyclic codes can be used to construct convolutional codes with large free distance.

\end{abstract}


\section{Introduction}\label{sec-into}


The theory of error-correction codes is an important research area, and it has many applications in modern life. For example, error-correction codes are widely used in cell phones to correct errors arising from fading noise during high frequency radio transmission. One of the major challenges in coding theory remains to construct new error-correction codes with good parameters and to study various problems such as finding the minimum distance and designing efficient decoding and encoding algorithms.

Let $\F$ be a finite field of order $q$. Let $n_1,n_2$ be two distinct odd primes such that $(n_1n_2,q)=1$ and $q$ is a quadratic residue for both $n_1$ and $n_2$. In an interesting paper \cite{Ding1}, Ding provided a general construction of cyclic codes of length $n_1n_2$ and dimension $(n_1n_2+1)/2$ over $\F$ by using generalised cyclotomies of order two in $\Z_{n_1n_2}^*$, the multiplicative subgroup of the ring $\Z_{n_1n_2}:=\Z/n_1n_2\Z$. This construction is reminiscent of that of quadratic residue codes of prime length, which can be defined by using cyclotomy of order two in $\Z_n^*$ when $n$ is a prime number. There are three different generalised cyclotomies of order two in $\Z_n^*$ when $n$ is a product of two distinct odd primes, corresponding to Ding's first, second and third construction, each of which yields 8 cyclic codes of length $n$ and dimension $(n+1)/2$. More interestingly, by using Magma and extensive numerical computation, Ding found some striking information on these 24 cyclic codes, which we list below in the table (see \cite{Ding1}):


\begin{center}
\begin{tabular}{|c||c||c|| c|}
\hline
\hline
$(n_1,n_2,q)$ & $(7,17,2)$ &$(11,13,3)$ & $(5,7,4)$\\
[0.5ex]
\hline
\hline
$\left[n_1n_2,\frac{n_1n_2+1}{2}\right]$ & $[119,60]$ & $[143,72]$ & $[35,18]$\\
\hline
\hline
Construction 1 &4 codes $d=12$* &4 codes $d=12$* &4 codes $d=8$*\\
&4 codes $d=11$ &4 codes $d=11$ &4 codes $d=7$ \\
\hline
Construction 2 &4 codes $d=12$* &4 codes $d=12$* &4 codes $d=7$\\
&4 codes $d=6$ &4 codes $d=6$ &4 codes $d=4$\\
\hline
Construction 3 &4 codes $d=8$ &4 codes $d=12$*&4 codes $d=8$*\\
&4 codes $d=4$ &4 codes $d=6$&4 codes $d<8$\\
\hline
\hline
\end{tabular}


\noindent {\small *: best cyclic codes of the same length \& dim over $\F$. Here $d$ is the min distance. }
\end{center}

The above data clearly demonstrates that Ding's three constructions of cyclic codes are promising and may produce many best cyclic codes. Hence it is worth studying these cyclic codes.

In this paper we provide a theory on Ding's three constructions that partially explains some mysteries of the above data: first, we prove that under permutation equivalence, there are indeed two codes in each construction; second, we prove an ``almost'' square-root bound on the minimum distance which is satisfied by all these codes; third, for Ding's second and third construction, we illustrate why half of the cyclic codes have relatively small minimal distance (see Theorems \ref{4:thm4}, \ref{4:thm5} and \ref{4:thm6} in Section \ref{sec4}). In order to study these codes, we develop a general method to study cyclic codes of composite length and to estimate the minimum distance (see Theorem \ref{2:thm1} in Section \ref{sec2}). This method seems new and may be useful for other problems in coding theory. Inspired by this method, we provide a general construction of cyclic codes of composite length which are also related to quadratic residue codes of prime length and study their properties (see Theorems \ref{3:thm2} and \ref{3:thm3} in Section \ref{sec3}). Numerical data shows that many of these codes are the best cyclic codes of the same length and dimension over $\F$ and hence are interesting. Finally, in light of \cite[Theorem 3]{Jus}, cyclic codes of composite length and the estimate of the minimal distance are very useful in the construction of convolutional codes with a prescribed free distance. As applications of results in the paper, we construct two families of convolutional codes with large free distance by cyclic codes from our construction and from Ding's constructions (Theorem \ref{5:thm9} in Section \ref{sec5}). In Section \ref{conclude} we conclude the paper and propose two open problems.



\section{Cyclic codes of composite length} \label{2:general cyclic code} \label{sec2}

In this section we study cyclic codes of composite length. We first introduce some standard notation which will be used throughout the paper.

Let $\F$ be the finite field of order $q$, where $q$ is a prime power. A linear $[n,k,d;q]$ code $\C$ is a $k$-dimensional subspace of $\F^n$ with minimum (Hamming) distance $d=d(\C)$. A linear $[n,k]$ code $\C$ over $\F$ is called a \emph{cyclic code} of length $n$ if any $(c_0,c_1,\ldots, c_{n-1}) \in \C$ implies $(c_{n-1},c_0,c_1,\ldots,c_{n-2}) \in \C$. By identifying any vector $(c_0,c_1,\ldots,c_{n-1}) \in \F^n$ with
\[c(x)=c_0+c_1x+\cdots+c_{n-1}x^{n-1} \in R_n:=\F[x]/(x^n-1),\]
it is known that $\C$ is a cyclic code of length $n$ over $\F$ if and only if the corresponding subset of $R_n$, still written as $\C$, is an idea of the ring $R_n$. Since every ideal of $R_n$ is principal, there is a monic polynomial $g(x) \in \F[x]$ of least degree such that $\C=(g(x)) \subset R_n$. This $g(x)$ is unique, satisfying $g(x)|(x^n-1)$ and is called the \emph{generator polynomial} of $\C$, and $h(x):=(x^n-1)/g(x)$ is called the \emph{parity-check polynomial} of $\C$.

Two codes $\C_1$ and $\C_2$ are called permutation equivalent, written as $\C_1 \sim \C_2$ if there is a permutation of coordinates that sends $\C_1$ to $\C_2$. The permutation of coordinates is called a permutation equivalence.

For any $c(x) \in \F[x]$, define $\supp(c(x))$ to be the set of integers $i$ such that the term $x^i$ appears in $c(x)$. Define the weight $\wt(c(x))$ to be the cardinality of the set $\supp(c(x))$. The main theorem is the following.

\begin{thm} \label{2:thm1}  Let $n,r \ge 2$ be positive integers such that $\gcd(nr,q) =\gcd(n,r)=1$. Assume that $r|(q-1)$. Let $\theta$ be a primitive $nr$-th root of unity in some extension of $\F$. Define $\lambda:=\theta^n$ and let $\bar{n}$ be a positive integer such that $n \bar{n} \equiv 1 \pmod{r}$. For any $0 \le t \le r-1$, let $\theta_t:=\lambda^{\bar{n}t}$. We define the map
\[ \phi:  \frac{\F[x]}{(x^{nr}-1)}  \longrightarrow  \left(\frac{\F[x]}{(x^n-1)}\right)^r\] by \[\phi:  c(x)  \mapsto \frac{1}{r} \left(\sum_{t=0}^{r-1} c_t(x) \lambda^{-tk}\right)_{k=0}^{r-1}\,\,,\]
here for any $c(x) \in \F[x]/(x^{nr}-1)$, the polynomial $c_t(x) \in \F[x]/(x^n-1)$ is given by
\[c_t(x) \equiv c(x\theta_t) \pmod{x^n-1} \quad \forall \, 0 \le t \le r-1.\]
Then:
\begin{itemize}
\item[1)] the map $\phi$ is a permutation equivalence, and $c(x) \ne 0$ if and only if $\left(c_t(x)\right)_{t=0}^{r-1} \ne 0$;

\item[2)] if $c(x) \ne 0$, then
\[ \wt(c(x)) \ge \min_t \left\{\wt(c_t(x)): c_t(x) \ne 0\right\}.\]
Moreover,
\begin{itemize}
\item[2.1)] if $c_0(x)=c_1(x)=\cdots=c_{r-1}(x)$, then $\wt(c(x))=\wt(c_0(x))$;

\item[2.2)] if $c_t(x)=0$ for some $t$, then
\[\wt(c(x)) \ge 2 \min_t \left\{\wt(c_t(x)): c_t(x) \ne 0\right\}.\]
\end{itemize}
\end{itemize}
Let $\C =(g(x)) \subset \F[x]/(x^{nr}-1)$ be a cyclic code with the generator polynomial $g(x)$. Then
\begin{itemize}
\item[3)] $\C$ is permutation equivalent to $\phi(C)$, which is given by
\[\phi(\C)=\left\{\left(\sum_{t=0}^{r-1} c_t(x) \lambda^{-tk}\right)_{k=0}^{r-1}: c_t(x) \in \C_t \quad \forall \, 0 \le t \le r-1\, \right\},\]
where $\C_t=(g_t(x)) \subset \F[x]/(x^n-1)$ is a cyclic code with the generator polynomial $g_t(x)$ given by
\[g_t(x)=\gcd\left(g(x \theta_t), x^n-1\right)\quad \forall \, 0 \le t \le r-1.\]

\item[4)] If $\C \ne 0$, then \[d(\C) \ge \min_t \left\{d(\C_t): C_t \ne 0 \right\}. \]

Moreover,
\begin{itemize}
\item[4.1)] if $\C_0=\C_1=\cdots=\C_{r-1}$, then $d(\C)=d(\C_0)$;

\item[4.2)] if $C_t =0$ for some $t$, then
\[d(\C) \ge 2 \min_t \left\{d(\C_t): C_t \ne 0 \right\}. \]
\end{itemize}
\end{itemize}
\end{thm}

\begin{proof} 1). Writting
\[c_t(x)=\sum_{z=0}^{n-1} c_{t,z}x^z, \quad c(x)=\sum_{i=0}^{nr-1} c_i x^i= \sum_{k=0}^{r-1}\sum_{z=0}^{n-1}c_{kn+z}x^{kn+z}, \]
then from $c_t(x) \equiv c(x \theta_t) \pmod{x^n-1}$, we find
\[c_{t,z}=\sum_{k=0}^{r-1} \sum_{z=0}^{n-1} c_{kn+z} \theta_t^{kn+z}, \quad \forall \, t, z. \]
Since $\theta_t=\lambda^{\bar{n}t}$ and $\lambda$ is a primitive $n$-th root of unity, we can obtain
\[c_{kn+z}=\frac{1}{r} \sum_{t=0}^{r-1} c_{t,z} \theta_t^{-kn-z}, \quad \forall \, t,z. \]
Thus we have
\[c(x)=\frac{1}{r} \sum_{t=0}^{r-1} \sum_{z=0}^{n-1} c_{t,z} x^z \sum_{k=0}^{r-1} x^{nk}\lambda^{-t(k+\bar{n} z)}. \]
For any given $z$, let $k' \equiv k+\bar{n}z \pmod{r}$. Noting that $x^{nk} \equiv x^{n(k'-\bar{n}z)} \pmod{x^{nr}-1}$ and as $k'$ runs over a complete residue system modulo $r$, so does $k \equiv k'-\bar{n}z \pmod{r}$, and it is clear that $\psi: x^{nk} \mapsto x^{nk'}$ induces a permutation of coordinates in $\F[x]/(x^{nr}-1)$,  thus $c(x)$ is permutation equivalent to
\begin{eqnarray} \label{2:psi} \psi(c(x))=\frac{1}{r} \sum_{t=0}^{r-1} \sum_{z=0}^{n-1} c_{t,z} x^z \sum_{k=0}^{r-1} x^{nk}\lambda^{-tk}=\frac{1}{r} \sum_{k=0}^{r-1} x^{nk}\sum_{t=0}^{r-1} c_t(x) \lambda^{-tk}.\end{eqnarray}
Hence $\psi(c(x))$ is permutation equivalent to $\phi(c(x))$, and thus $\phi$ is a permutation equivalence. Moreover, noting
\[ c_t(x)=0 \Longleftrightarrow (x^n-1)|c_t(x)  \Longleftrightarrow (x^n-\lambda^t)|c(x),\]
and
\[x^{nr}-1=\prod_{t=0}^{r-1} (x^n-\lambda^t),\]
it is obvious that $c(x) \ne 0$ if and only if $\left(c_t(x)\right)_{t=0}^{r-1} \ne 0$. This proves 1).

2). Since $c_t(x) \equiv c(x \theta_t) \pmod{x^n-1}$, we have
\[\wt(c(x))=\wt(c(x \theta_t)) \ge \wt(c_t(x)) \, \forall \, 0 \le t \le r-1. \]
Taking $c_0(x)=\cdots =c_{r-1}(x)$ in (\ref{2:psi}), and using
\[\sum_{t=0}^{r-1} \lambda^{-tk}=\left\{ \begin{array}{lll}
r&:& \mbox{ if } r|k\\
0&:& \mbox{ if } r \nmid k \end{array}\right., \]
we find easily that $\psi(c(x))=c_0(x)$. This proves 2) and 2.1).

As for 2.2), let
\[A=\left\{0 \le t \le r-1: c_t(x) \ne 0 \right\}, \quad \I=\bigcup_{t \in A} \supp(c_t(x)).\]
From (\ref{2:psi}) we find
\[\wt(c(x))=\sum_{z \in \I} \sum_{k=0}^{r-1} \wt \left(\sum_{t \in A} c_{t,z} \lambda^{-tk}\right). \]
For each $k$ and $z$, write
\[h_{k,z}:=\sum_{t \in A} c_{t,z} \lambda^{-tk}, \]
and for each $z$, write $\vec{h}_z:=(h_{k,z})_{k=0}^{r-1}$ and $\vec{c}_z:=(c_{t,z})_{t \in A}$ as column vectors. Assume that $A=\{t_1,t_2,\ldots,t_u\}$. Then we have
\begin{eqnarray} \label{2:matrix} \left[\begin{array}{cccc}
1&1& \cdots &1\\
\lambda^{-t_1}&\lambda^{-t_2}& \cdots &\lambda^{-t_u}\\
\vdots&\vdots& \vdots &\vdots\\
\lambda^{-(r-1)t_1}&\lambda^{-(r-1)t_2}&\cdots &\lambda^{-(r-1)t_u}\end{array}\right] \cdot \vec{c}_z=\vec{h}_z,\end{eqnarray}
and
\[\wt(c(x))=\sum_{z \in \I} \wt(\vec{h}_z).\]
Noting that for any $z \in \I$, we have $\vec{c}_z \ne 0$, and the matrix on the left side of the equation (\ref{2:matrix}) is a Vandermond matrix of size $r \times u$ where $1 \le u<r$, we find $\wt(\vec{h}_z) \ge 2$. Thus $\wt(c(x)) \ge \sum_{z \in \I} 2 \ge 2 \min\limits_t \{\wt(c_t(x)): c_t(x) \ne 0\}$. This proves 2.2).

3). For any $0 \le t \le r-1$, define the maps $\Phi_t, f_t$
\[\frac{\F[x]}{(x^{nr}-1)} \overset{\Phi_t}{\longrightarrow} \frac{\F[x]}{(x^{n}-\lambda^t)} \overset{f_t}{\longrightarrow} \frac{\F[x]}{(x^{n}-1)}\]
by
\[\Phi_t: c(x) \mapsto c(x) \pmod{x^n-\lambda^t}, \quad f_t: x \mapsto x \theta_t. \]
Let $\Psi_t:=f_t \circ \Phi_t$ and $\Psi=(\Psi_t)_{t=0}^{r-1}$. The isomorphism from the Chinese Remainder Theorem
\[\Psi: \frac{\F[x]}{(x^{nr}-1)} \longrightarrow \prod_{t=0}^{r-1} \frac{\F[x]}{(x^{n}-1)}\]
induces an isomorphism $\Psi(\C) \cong \prod\limits_{t=0}^{r-1} \Psi_t(\C)$. Since clearly $\C_t=\Psi_t(\C)$, hence 3) is proved.

Finally, 4), 4.1) and 4.2) are direct consequences of 2)--2.2). This completes the proof of Theorem \ref{2:thm1}.

\end{proof}

\begin{xrmk} In Theorem \ref{2:thm1}, the requirement $r|(q-1)$ is mild and can be removed as follows: suppose $\C =(g(x)) \subset \F[x]/(x^{nr}-1)$ is a cyclic code and $r \nmid (q-1)$. There is a $k$ such that $r|(q^k-1)$. Let $q'=q^k$ and $\bar{\C}=(g(x)) \subset \FF_{q'}[x]/(x^{nr}-1)$ be the extension of $\C$ over $\FF_{q'}$. It is known that $\wt(\C)=\wt(\bar{\C})$ (\cite[Theorem 3.8.8]{Huf}). Hence to study the minimum distance of $\C$ over $\F$, it is equivalent to study $\bar{\C}$ on which the condition $r|(q'-1)$ is satisfied.
\end{xrmk}

\section{Cyclic codes of composite length related to quadratic residue codes} \label{sec3}

Inspired by Theorem \ref{2:thm1}, we give a construction of cyclic codes of composite length which are closely related to quadratic residue codes of prime length.

Throughout this section, we make the following assumptions:
\begin{itemize}
\item[3.1)] $\F$ is the finite field of order $q$,
\item[3.1)] $n \ge 3$ is a prime number such that $\legn{q}{n}=1$ where $\legn{\cdot}{n}$ is the Legendre symbol for the prime $n$,
\item[3.3)] $r \ge 2$ is a positive integer with $\gcd(nr,q)=\gcd(n,r)=1$.
\end{itemize}

Let $\theta$ be a primitive $nr$-th root of unity in some extension of $\F$ and let $\lambda:=\theta^n$. For each $0 \le t \le r-1$, we may factor $x^n-\lambda^t$ as
\[x^n-\lambda^t=\prod_{j=1}^{n} \left(x-\theta^{t+rj}\right). \]
For each $t$, there is a unique integer $j_t$ where $1 \le j_t \le n$ such that $t+rj_t \equiv 0 \pmod{n}$. Define $\theta_t:=\theta^{t+rj_t}$. It is easy to see that $\theta_t=\lambda^{\bar{n}t}$ where $\bar{n}$ is a positive integer such that $\bar{n}n \equiv 1 \pmod{r}$. We can write
\[x^n-\lambda^t=(x-\theta_t) g_{t,1}(x) g_{t,-1}(x),\]
where for $\epsilon \in \{\pm 1\}$,
\[g_{t,\epsilon}(x)=\prod_{\substack{1 \le j \le n\\
\legn{t+rj}{n}=\epsilon}} \left(x-\theta^{t+rj}\right).\]
For any $\vec{\epsi}=(\epsi_0,\epsi_1,\dots, \epsi_{r-1}) \in \{\pm 1\}^r$, define
\begin{eqnarray} \label{3:g} g_{\vec{\epsi}}(x)=\prod_{t=0}^{r-1} g_{t,\epsi_t}(x). \end{eqnarray}
Let
\begin{eqnarray} \label{3:Lambda} \Lambda:=\left\{\vec{\epsi}=(\epsi_0,\epsi_1,\cdots,\epsi_{r-1}) \in \{\pm 1\}^r: \epsi_j=\epsi_{qj} \,\,\, \forall \, 0 \le j \le r-1\right\}. \end{eqnarray}
Here the subscript $qj$ in $\epsi_{qj}$ is always understood to be $qj \pmod{r}$. For any $\vec{\epsi} \in \Lambda$, since $\legn{q}{n}=1$, it is easy to see that $g_{\vec{\epsi}}(x) \in \F[x]$.

We remark here that a defining set similar to $\Lambda$ has been used in \cite{Ding0} to construct binary duadic codes of prime length.

\begin{defn} Assume 3.1)--3.3). We define $\C_{\vec{\epsi}}$ to be the cyclic code of length $nr$ over $\F$ with the generator polynomial $g_{\vec{\epsi}}(x)$ given in (\ref{3:g}), that is,
\begin{eqnarray} \label{3:defn1}
\C_{\vec{\epsi}}=\left(g_{\vec{\epsi}}(x)\right) \subset \F[x]/(x^{nr}-1).
\end{eqnarray}
\end{defn}

\begin{xrmk}
The code $\C_{\vec{\epsi}}$ in (\ref{3:defn1}) is a cyclic code over $\F$ of length $nr$ and dimension $(n+1)r/{2}$. The total number of these codes is $2^{\tau(q,r)}$, where $\tau(q,r):=\#\Lambda$ is the number of $q$-cyclotomic cosets modulo $r$. In particular there are always at least 4 such codes since $\tau(q,r) \ge 2$ for any $r \ge 2$, and if $q \equiv 1 \pmod{r}$, then $\tau(q,r)=r$, hence in this case the number of such codes is $2^r$. In general there may be many such cyclic codes $\C_{\vec{\epsi}}$.
\end{xrmk}
We first study the codes $\C_{\vec{\epsi}}$ under permutation equivalence. For any $\vec{\epsi} =(\epsi_0,\epsi_1,\cdots,\epsi_{r-1})\in \Lambda$, define $-\vec{\epsi}:=(-\epsi_0,-\epsi_1,\cdots,-\epsi_{r-1})$. For any $u$ with $\gcd(u,r)=1$, define $\vec{\epsi}^{*u}:=(\epsi_t^{*u})_{t=0}^{r-1}$ as
\[\epsi_t^{*u}:= \epsi_{tu} \quad \forall \, 0 \le t \le r-1. \]
We have the following:

\begin{thm} \label{3:thm2} Assume 3.1)--3.3). For any $\vec{\epsi} \in \Lambda$ and any 
$u$ such that $\gcd(u,r)=1$, the codes $\C_{\vec{\epsi}}, \C_{-\vec{\epsi}}$ and $\C_{\vec{\epsi}^{*u}}$ given in (\ref{3:defn1}) are all permutation equivalent.
In particular we have
\[ d \left(\C_{\vec{\epsi}}\right)= d \left(\C_{-\vec{\epsi}}\right)= d \left(\C_{\vec{\epsi}^{*u}}\right).\] 

\end{thm}
\begin{proof} For $\C_{\vec{\epsi}}$, by Theorem \ref{2:thm1}, we find that
\[g_t(x)=\gcd\left(g_{\vec{\epsi}}(x \theta_t),x^n-1\right)=g_{0,\epsi_t}(x), \quad \forall \, 0 \le t \le r-1. \]
Hence the code $\C_{\vec{\epsi}}$ is permutation equivalent to
\[\phi(\C_{\vec{\epsi}})=\left\{ \left(\sum_{t=0}^{r-1} c_t(x) \lambda^{-tk}\right)_{k=0}^{r-1}: c_t(x) \in \C_{0,\epsi_t} \, \forall \, 0 \le t \le r-1\right\},\]
where
\begin{eqnarray} \label{3:quadratic} \C_{0,\epsi_t} :=\left(g_{0,\epsi_t}(x)\right) \subset \F[x]/(x^n-1)\end{eqnarray}
is a quadratic residue code of length $n$ and dimension $(n+1)/2$ over $\F$. Let $s$ be a positive integer such that $\legn{s}{n}=-1$. It is known that the map $x \mapsto x^s$ induces a permutation equivalence from $\F[x]/(x^n-1)$ to itself and $c_t(x) \in \C_{0,\epsi_t}$ if and only if $c_t(x^s) \in \C_{0,-\epsi_t}$. Thus by taking $x \mapsto x^s$ we find that $\phi(\C_{\vec{\epsi}})$ is permutation equivalent to $\phi(\C_{-\vec{\epsi}})$.

For any $c(x) \in \C_{\vec{\epsi}}$,
\[\phi(c(x))=\left(\sum_{t=0}^{r-1} c_t(x) \lambda^{-tk}\right)_{k=0}^{r-1}=\left(\sum_{t=0}^{r-1} c_{ut}(x) \lambda^{-utk}\right)_{k=0}^{r-1},\]
where $c_{ut}(x) \in \C_{o,\epsi_{ut}}$, and the subscript $ut$ is understood to be $ut \pmod{r}$. It is easy to see that the map $k \mapsto uk$ induces a permutation equivalence between $\phi(c(x))$ and
\[\left(\sum_{t=0}^{r-1} c_{ut}(x) \lambda^{-tk}\right)_{k=0}^{r-1}.\]
Hence $\C_{\vec{\epsi}}$ and $\C_{\vec{\epsi}^{*u}}$ are permutation equivalent. This proves Theorem \ref{3:thm2}.


\end{proof}

Now we can estimate the minimum distance of these codes $\C_{\vec{\epsi}}$. Let $d_{n,q}$ be the minimum distance of the quadratic residue code $\C_{0,\epsi}$ given in (\ref{3:quadratic}) ($\epsi=\pm 1$). It is known that $d_{n,q} > \sqrt{n}$ (\cite{Huf}).

\begin{thm} \label{3:thm3} Assume 3.1)--3.3). For any $\vec{\epsi} \in \Lambda$, the code $\C_{\epsi}$ given in \ref{3:defn1} is an $\left[nr,\frac{(n+1)r}{2} \right]$ code with the following properties:
\begin{itemize}
\item[1)] $d\left(\C_{\vec{\epsi}}\right) \ge d_{n,q}$;
\item[2)] if $\vec{\epsi}=(0,0,\cdots,0)$ or $(1,1,\cdots,1)$, then $d\left(\C_{\vec{\epsi}}\right) = d_{n,q}$\,;

\item[3)] if $\vec{\epsi} \ne (0,0,\cdots,0)$ or $(1,1,\cdots,1)$, then
\[d\left(\C_{\vec{\epsi}}\right) \ge \min \left\{d_{n,q}+1, \frac{\sqrt{8n+1}-1}{2}\right\}. \]

\end{itemize}
\end{thm}

\begin{proof} 1) and 2) are direct consequences of Theorem \ref{2:thm1}, noting that for any $t$,
\[g_t(x)=\gcd\left(g_{\vec{\epsi}}(x \theta_t),x^n-1\right)=g_{0,\epsi_t}(x), \]
and
\[\C_t=(g_{0,\epsi_t}(x)) \subset \F[x]/(x^n-1), \quad d \left(\C_t\right)=d_{n,q}. \]

As for 3), suppose $d\left(\C_{\vec{\epsi}}\right)=d_{n,q}$. Let $c(x) \in \C_{\vec{\epsi}}$ with $\wt(c(x))=d_{n,q}$. Then by the proof of Theorem \ref{2:thm1}, the corresponding $\left(c_t(x)\right)_{t=0}^{r-1}$ satisfies
\begin{itemize}
\item[i)] $c_t(x) \ne 0$ for any $t$,

\item[ii)] let $\I=\bigcup\limits_{t=0}^{r-1} \supp(c_t(x))$, then $\#\I = d_{n,q}$.
\end{itemize}
So for each $t$, we have $0 \ne c_t(x)=\sum_{z \in \I}c_{t,z} x^z \in \C_{0,\epsi_t}$ with $\wt(c_t(x))=d_{n,q}$ being the minimum weight in a quadratic residue code of length $n$ and dimension $(n+1)/2$. It is known that $c_t(1) \ne 0$ for all $t$. There are $t_1,t_2$ such that
\[c_{t_1}(x) \in \C_{0,1}, \quad c_{t_2}(x) \in \C_{0,-1}. \]
So we have
\[\left.\frac{x^n-1}{x-1}\right| c_{t_1}(x)c_{t_2}(x) \Longrightarrow \wt(c_{t_1}(x)c_{t_2}(x))=n. \]
On the other hand,
\[c_{t_1}(x)c_{t_2}(x)=\sum_{z_1,z_2 \in \I} c_{t_1,z}c_{t_2,x} x^{z_1+z_2} \pmod{x^n-1}. \]
We have
\[n=\wt(c_{t_1}(x)c_{t_2}(x)) \le \# (\I+\I) \le d_{n,q}+\binom{d_{n,q}}{2}=\frac{d_{n,q}(d_{n,q}+1)}{2}. \]
This implies that
\[d_{n,q} \ge \frac{\sqrt{8n+1}-1}{2}.\]
This completes the proof of Theorem \ref{3:thm3}.
\end{proof}

\begin{exam} In examples below, we use Magma to compute the parameters of some codes for $q=2,3,4$ and a few small values $n,r$. We only consider codes which are not permutation equivalent by Theorem \ref{3:thm2}. The $^*$-symbol indicates that the code is the best among all cyclic codes, $^\checkmark$-symbol indicates that the code is the best among all linear codes, and $^\diamond$-symbol indicates that the code is optimal among all codes. For $q=2,3$, such information is obtained by consulting \cite[Appendix A: Tables of best binary and ternary cyclic codes]{Ding2}; for $q=4$ we consult the online reference \emph{http://www.codetables.de/} to check the best linear codes. There are also some sporadic examples of best cyclic codes that we exclude from the table. It demonstrates that the construction produces many best cyclic codes, and when the parameters are small, it may even produce best linear codes.
\begin{center}
\begin{tabular}{|c||c||c|| c||c|}
\hline
\hline
$q=2$&$n=7$ & $n=17$ &$n=23$ & $n=31$ \\
[0.5ex]
\hline
\hline
$r=3$&$[21,12]$ code &$[51,27]$ code  & $[69,36]$ code &$[93,48]$ code  \\
&1 code $d=5^\diamond$ &1 code $d=9^*$  & 1 code $d=11^*$ & 1 code $d=14^\checkmark$  \\
&1 code $d=3$ &1 code $d=5$  & 1 code $d=7$ & 1 code $d=7$  \\
\hline
\hline
$r=5$&$[35,20]$ code &$[85,45]$ code  &$[115,60]$ code & $[155,80]$ code  \\
&1 code $d=6^*$ &1 code $d=10$  & 1 code $d=14^*$ &1 code $d=14$  \\
&1 code $d=3$ &1 code $d=5$  & 1 code $d=7$ &1 code $d=7$  \\
\hline
\hline
$r=7$&NA &$[119,63]$ code  & $[161,84]$ code & $[217,112]$ code \\
& &1 code $d=11^*$  &2 codes $d=14$ &  NA \\
& &1 code $d=10$  &1 code $d=7$ &  \\
& &1 code $d=5$  & & \\
\hline
\hline
\end{tabular}
\label{1:t21}
\end{center}

\begin{center}
\begin{tabular}{|c||c||c|| c||c|}
\hline
\hline
$q=3$&$n=11$ & $n=13$ &$n=23$ & $n=37$ \\
[0.5ex]
\hline
\hline
$r=2$&$[22,12]$ code &$[26,14]$ code  &$[46,24]$ code &$[74,38]$ code  \\
& 1 code $d=7^\diamond$ &1 code $d=7^\checkmark$  &1 code $d=13^\checkmark$ &1 code $d=14^*$  \\
& 1 code $d=5$ &1 code $d=5$  &1 code $d=8$ &1 code $d=10$  \\
\hline
\hline
$r=4$&$[44,24]$ code &$[52,28]$ code  &$[92,48]$ code &$[148,76]$ code  \\
& 2 codes $d=8^*$ &2 codes $d=9$  &2 codes $d=14$ & NA  \\
& 1 code $d=7$ &1 code $d=7$  &1 code $d=13$ &\\
& 1 code $d=5$ &1 code $d=5$  &1 code $d=8$ & \\
\hline
\hline
\end{tabular}
\label{1:t22}
\end{center}

\begin{center}
\begin{tabular}{|c||c||c|| c||c||c|}
\hline
\hline
$q=4$&$n=5$ & $n=7$ &$n=11$ & $n=13$ & $n=17$\\
[0.5ex]
\hline
\hline
$r=3$&$[15,9]$ code &$[21,12]$ code &$[33,18]$ code &$[39,21]$ code &$[51,27]$ code\\
& 2 codes $d=5^\checkmark$ & 2 codes $d=5$  & 2 codes $d=8$ & 2 codes $d=9$  &  2 codes $d=9$\\
& 1 code $d=3$ & 1 code $d=3$  & 1 code $d=5$ & 1 code $d=5$  &  1 code $d=5$\\
\hline
\hline
\end{tabular}

\label{1:t23}
\end{center}
\end{exam}

\section{Cyclic codes from Ding's constructions} \label{sec4}

Throughout this section we assume that
\begin{itemize}
\item[4.1)] $\F$ is the finite field of order $q$,
\item[4.2)] $n_1,n_2$ are two distinct odd primes such that $\legn{q}{n_1}=\legn{q}{n_2}=1$ and $\gcd(n_1n_2,q)=1$,
\item[4.3)] $\theta$ is a primitive $n_1n_2$-th root of unity in some extension of $\F$, and $\lambda:=\theta^{n_1}$.
\end{itemize}

Ding's three constructions of cyclic codes can be described explicitly by using the three different generalized cyclotomies of order two in $\Z_{n_1n_2}^*$ (see \cite{Ding1}). However, it may look more natural to describe these constructions as follows.

\subsection{Ding's $1^{\mathrm st}$ construction}

We write
\[x^{n_1n_2}-1=\prod_{i=0}^{n_1n_2-1}\left(x-\theta^i\right)=(x-1) F^{(1)}_{1,1}(x)F^{(1)}_{1,-1}(x)F_{2,1}(x)F_{2,-1}(x)F_{3,1}(x)F_{3,-1}(x),\]
where for $\epsi \in \{\pm 1\}$,
\begin{eqnarray} \label{4:con0} F^{(1)}_{1,\epsi}(x)=\prod_{\substack{1 \le i \le n_1n_2-1\\
\legn{i}{n_1n_2}=\epsi}} \left(x-\theta^i\right),\end{eqnarray}
here $\legn{\cdot}{n_1n_2}=\legn{\cdot}{n_1} \legn{\cdot}{n_2}$ is the Jacobi symbol for $n_1n_2$, and
\begin{eqnarray}
\label{4:f2}
F_{2,\epsi}(x)&=&\prod_{\substack{1 \le i \le n_2-1\\
\legn{n_1i}{n_2}=\epsi}} \left(x-\theta^{n_1i}\right),\\
\label{4:f3}
F_{3,\epsi}(x)&=&\prod_{\substack{1 \le i \le n_1-1\\
\legn{n_2i}{n_1}=\epsi}} \left(x-\theta^{n_2i}\right). \end{eqnarray}
Since $\legn{q}{n_1}=\legn{q}{n_2}=1$, we have $F^{(1)}_{1,\epsi}(x), F_{2,\epsi} (x), F_{3,\epsi}(x) \in \F[x]$ for all $\epsi$.

Denote by $d_{n_1,q}$ the minimum distance of the code $\left(F_{3,\epsi}(x) \right) \subset \F[x]/(x^{n_1}-1)$. This is a quadratic residue code of length $n_1$ and dimension $(n_1+1)/2$ over $\F$. Denote by $\bar{d}_{n_1,q}$ the minimum distance of the code $\left((x-1)F_{3,\epsi}(x) \right) \subset \F[x]/(x^{n_1}-1)$. It is known that
\[ \bar{d}_{n_1,q}=d_{n_1,q}+1 \quad \mbox{ and } \quad d_{n_1,q} > \sqrt{n_1}.\]
We also define $d_{n_2,q}$ and $\bar{d}_{n_2,q}$ similarly.

\begin{defn}[Ding's $1^{\mathrm st}$ constructions] Assume 4.1)--4.3). For any $\vec{\epsi}:=(\epsi_1,\epsi_2,\epsi_3) \in \{\pm 1\}^3$, let
 \begin{eqnarray*} \label{4:d1g} g^{(1)}_{\vec{\epsi}}(x):=F^{(1)}_{1,\epsi_1}(x)F_{2,\epsi_2}(x)F_{3,\epsi_3}(x).\end{eqnarray*}
Then define $C^{(1)}_{\vec{\epsi}}$ to be the cyclic code of length $n_1n_2$ over $\F$ with the generator polynomial $g^{(1)}_{\vec{\epsi}}(x)$, that is,
\begin{eqnarray} \label{4:defn1} C^{(1)}_{\vec{\epsi}}=\left(g^{(1)}_{\vec{\epsi}}(x)\right) \subset \F[x]/(x^{n_1n_2}-1). \end{eqnarray}
\end{defn}
Next we state the main results on these codes $\cone{1}$.

\begin{thm}  \label{4:thm4}
\begin{itemize}
\item[i)] For any $\vepsi \in \{\pm 1\}^3$, the cyclic code $\cone{1}$ over $\F$ given in (\ref{4:defn1}) has  length $n_1n_2$ and dimension $(n_1n_2+1)/2$, satisfying
\[d\left(\cone{1}\right) > \max \left\{d_{n_1,q},d_{n_2,q}\right\} > \max \left\{\sqrt{n_1},\sqrt{n_2}\right\}. \]

\item[ii)] Write $\{\pm 1\}^3=A_1 \cup A_2$ where
\begin{eqnarray*} A_1&=&\{(1,1,1),(-1,-1,1),(-1,1,-1),(1,-1,-1)\}, \\
A_2&=&\{(-1,-1,-1),(1,1,-1),(1,-1,1),(-1,1,1)\}.\end{eqnarray*}
Then for $i=1$ or $2$, the codes $\cone{1}$ with $\vepsi \in A_i$ are all permutation equivalent.
\end{itemize}
\end{thm}

\begin{proof}

i). Obviously the code $C^{(1)}_{\vec{\epsi}}$ is of length $n_1n_2$ and dimension $(n_1n_2+1)/2$. To study the minimum distance, we may assume that $n_1n_2|(q-1)$.

Applying Theorem \ref{2:thm1} for $n=n_1$ and $r=n_2$, we find that $C^{(1)}_{\vec{\epsi}}$ is permutation equivalent to
\[\phi\left(C^{(1)}_{\vec{\epsi}}\right)=\left\{\left(\sum_{t=0}^{n_2-1}c_t(x) \lambda^{-tk}\right)_{k=0}^{n_2-1}: c_t(x) \in \C_t \, \,\, \forall \, 0 \le t \le n_2-1\right\}, \]
where for each $t$, $\C_t=(g_t(x)) \subset \F[x]/(x^{n_1}-1)$ is a cyclic code with
\begin{eqnarray*} \label{4:gt1} g_t(x)=\gcd \left(\gone{1}(x \theta_t),x^{n_1}-1\right)=\left\{\begin{array}{lll}g_{0,\epsi_3}(x)&:& t=0\\
(x-1)g_{0,\epsi_1\epsi_2}(x)&:& \legn{t}{n_2}=\epsi_2\\
g_{0,-\epsi_1\epsi_2}(x)&:& \legn{t}{n_2}=-\epsi_2.
\end{array}\right.\end{eqnarray*}
Here $\theta_t=\lambda^{\bar{n}_1t}$ where $\bar{n}_1$ is a positive integer such that $n_1 \bar{n}_1 \equiv 1 \pmod{n_2}$, and $g_{0,\epsi}(x)=F_{3,\epsi}(x)$ for any $\epsi$. It implies from Theorem \ref{2:thm1} that $d\left(\cone{1}\right) \ge d_{n_1,q} \ge \sqrt{n_1}$. A closer look from Theorem \ref{2:thm1} easily shows that the equality ``='' can not be achieved, hence we have $d\left(\cone{1}\right) > d_{n_1,q}$.

On the other hand, we may also apply Theorem \ref{2:thm1} for $n=n_2$ and $r=n_1$ to obtain $d\left(\cone{1}\right) > d_{n_2,q} > \sqrt{n_2}$. Thus we obtain
\[d\left(\cone{1}\right) > \max \left\{d_{n_1,q},d_{n_2,q}\right\} > \max \left\{\sqrt{n_1},\sqrt{n_2}\right\}. \]

ii). For any positive integer $u$ such that $\gcd(u,n_1n_2)=1$, the map $x \mapsto x^u$ induces a permutation equivalence $\phi_u: \F[x]/(x^{n_1n_2}-1) \to \F[x]/(x^{n_1n_2}-1)$. Choose two positive integers $a,b$ such that
\[\legn{a}{n_1}=-1, \quad \legn{a}{n_2}=1, \quad \legn{b}{n_1}=1, \quad \legn{b}{n_2}=-1,\]
and for any $\vepsi =(\epsi_1,\epsi_2,\epsi_3) \in \{\pm 1\}$, let
\[ \vepsi^{(a)}:=(-\epsi_1,\epsi_2,-\epsi_3), \quad \vepsi^{(b)}:=(-\epsi_1,-\epsi_2,\epsi_3).\]
It is easy to see that
\[\phi_a \left(C^{(1)}_{\vec{\epsi}}\right)= C^{(1)}_{\vec{\epsi}^{(a)}}, \quad \phi_b \left(C^{(1)}_{\vec{\epsi}}\right)= C^{(1)}_{\vec{\epsi}^{(b)}}. \]
So the codes $\cone{1}, \ctwo{1}{a}$ and $\ctwo{1}{b}$ are all permutation equivalent.

Considering these two actions $\vepsi \mapsto \vepsi^{(a)}$ and $\vepsi \mapsto \vepsi^{(b)}$, it is easy to check that the codes $\cone{1}$ are permutation equivalent for all $\vepsi \in A_i$ where $i=1$ or $2$. This completes the proof of Theorem \ref{4:thm4}.

\end{proof}

\subsection{Ding's $2^{\mathrm nd}$ construction}

We write
\[x^{n_1n_2}-1=\prod_{i=0}^{n_1n_2-1}\left(x-\theta^i\right)=(x-1) F^{(2)}_{1,1}(x)F^{(2)}_{1,-1}(x)F_{2,1}(x)F_{2,-1}(x)F_{3,1}(x)F_{3,-1}(x),\]
where for any $\epsi \in \{\pm 1\}$,
\begin{eqnarray}
\label{4:con1}
F^{(2)}_{1,\epsi}(x)&=&\prod_{\substack{1 \le i \le n_1n_2-1\\
\legn{i}{n_1}=\epsi\\
\gcd(i,n_2)=1}} \left(x-\theta^i\right),
\end{eqnarray}
and the functions $F_{2,\epsi}(x)$ and $F_{3,\epsi}(x)$ are defined in (\ref{4:f2}) and (\ref{4:f3}) respectively.

\begin{defn}[Ding's $2^{\mathrm nd}$ constructions] Assume 4.1)--4.3). For any $\vec{\epsi}:=(\epsi_1,\epsi_2,\epsi_3) \in \{\pm 1\}^3$, let
 \begin{eqnarray*} \label{4:d2g} g^{(2)}_{\vec{\epsi}}(x):=F^{(2)}_{1,\epsi_1}(x)F_{2,\epsi_2}(x)F_{3,\epsi_3}(x).\end{eqnarray*}
Then define $C^{(2)}_{\vec{\epsi}}$ to be the cyclic code of length $n_1n_2$ over $\F$ with the generator polynomial $g^{(2)}_{\vec{\epsi}}(x)$, that is,
\begin{eqnarray} \label{4:defn2}  C^{(2)}_{\vec{\epsi}}=\left(g^{(2)}_{\vec{\epsi}}(x)\right) \subset \F[x]/(x^{n_1n_2}-1). \end{eqnarray}
\end{defn}
Next we state the main results on these codes $\cone{2}$.

\begin{thm}  \label{4:thm5}
\begin{itemize}
\item[i)] For any $\vepsi \in \{\pm 1\}^3$, the cyclic code $\cone{2}$ over $\F$ given in (\ref{4:defn2}) has length $n_1n_2$ and dimension $(n_1n_2+1)/2$, satisfying
\[d\left(\cone{2}\right) > d_{n_1,q} > \sqrt{n_1}. \]

\item[ii)] Write $\{\pm 1\}^3=A_1 \cup A_2$ where
\begin{eqnarray*} A_1&=&\{(1,1,1),(1,-1,1),(-1,1,-1),(-1,-1,-1)\}, \\
A_2&=&\{(1,1,-1),(1,-1,-1),(-1,1,1),(-1,-1,1)\}. \end{eqnarray*}
Then for $i=1$ or $2$, the codes $\cone{2}$ for $\vepsi \in A_i$ are all permutation equivalent.

\item[iii)] Taking $\vepsi_1=(1,1,1) \in A_1$, then
\[ d\left(\C^{(2)}_{\vepsi_1}\right) = {d}_{n_1,q}+1.\]

\item[iv)] If $n_1 \equiv n_2 \equiv -1 \pmod{4}$, and $q$ is a prime power such that $q=2^m$ or $q \equiv 1 \pmod{4}$, let $\gamma \in \F$ such that
\begin{eqnarray} \label{4:duad} 1+\gamma^2 n_1n_2=0. \end{eqnarray}
Then for any $\vepsi \in \{\pm 1\}^3$, the code $\bar{\mathcal{C}}_{\vepsi}^{(2)}$ given by
\[\bar{\mathcal{C}}_{\vepsi}^{(2)}:=\left\{(c_0,\ldots,c_{n_1n_2-1},c_{\infty}): c_{\infty}=-\gamma \, c(1), \, \, c(x)=\sum_{i=0}^{n_1n_2-1} c_ix^i\in \cone{2} \right\}\]
has parameters $\left[n_1n_2+1,\frac{n_1n_2+1}{2}\right]$ and is self-dual.
\end{itemize}
\end{thm}

\begin{proof}

i). Obviously the code $\C^{(2)}_{\vec{\epsi}}$ is of length $n_1n_2$ and dimension $(n_1n_2+1)/2$. To study the minimum distance, we may assume that $n_1n_2|(q-1)$.

Applying Theorem \ref{2:thm1} for $n=n_1$ and $r=n_2$, we find that $C^{(2)}_{\vec{\epsi}}$ is permutation equivalent to
\[\phi\left(C^{(2)}_{\vec{\epsi}}\right)=\left\{\left(\sum_{t=0}^{n_2-1}c_t(x) \lambda^{-tk}\right)_{k=0}^{n_2-1}: c_t(x) \in \C_t \, \forall \, 0 \le t \le n_2-1\right\}, \]
where for each $t$, $\C_t=(g_t(x)) \subset \F[x]/(x^{n_1}-1)$ is a cyclic code with
\begin{eqnarray*} \label{4:gt21} g_t(x)=\gcd \left(\gone{1}(x \theta_t),x^{n_1}-1\right)=\left\{\begin{array}{lll}g_{0,\epsi_3}(x)&:& t=0\\
(x-1)g_{0,\epsi_1}(x)&:& \legn{t}{n_2}=\epsi_2\\
g_{0,\epsi_1}(x)&:& \legn{t}{n_2}=-\epsi_2.
\end{array}\right.\end{eqnarray*}
It implies from Theorem \ref{2:thm1} that $d\left(\cone{1}\right) > d_{n_1,q} > \sqrt{n_1}$.

ii). The permutation equivalence of $\cone{2}$ can be proved in exactly the same way as that of $\cone{1}$, hence we omit the details. We are contented by stating that for any $\vepsi=(\epsi_1,\epsi_2,\epsi_3) \in \{\pm 1\}^3$, define
\[\vepsi^{(a)}:=(-\epsi_1,\epsi_2,-\epsi_3), \quad \vepsi^{(b)}:=(\epsi_1,-\epsi_2,\epsi_3),\]
then the codes $\cone{2}$, $\ctwo{2}{a}$ and $\ctwo{2}{b}$ are all permutation equivalent. Considering these two actions $\vepsi \mapsto \vepsi^{(a)}$ and $\vepsi \mapsto \vepsi^{(b)}$, ii) is easily verified.

iii) Let $\vepsi=(1,1,1)$. Applying Theorem \ref{2:thm1} for $n=n_1$ and $r=n_2$, we find that $C^{(2)}_{\vec{\epsi}}$ is permutation equivalent to
\[\phi\left(C^{(2)}_{\vec{\epsi}}\right)=\left\{\left(\sum_{t=0}^{n_2-1}c_t(x) \lambda^{-tk}\right)_{k=0}^{n_2-1}: c_t(x) \in \C_t \, \forall \, 0 \le t \le n_2-1\right\}, \]
where for each $t$, $\C_t=(g_t(x)) \subset \F[x]/(x^{n_1}-1)$ is a cyclic code with
\begin{eqnarray*} \label{4:gt22} g_t(x)=\gcd \left(\gone{2}(x \theta_t),x^{n_1}-1\right)=\left\{\begin{array}{lll}g_{0,1}(x)&:& t=0\\
(x-1)g_{0,1}(x)&:& \legn{t}{n_2}=1\\
g_{0,1}(x)&:& \legn{t}{n_2}=-1.
\end{array}\right.\end{eqnarray*}
We may choose $c'(x) \in \left((x-1)g_{0,1}(x)\right) \subset \F[x]/(x^{n_1}-1)$ with $\wt(c'(x))=\bar{d}_{n_1,q}=d_{n,q}+1$, and let $c_t(x)=c'(x)$ for all $t$. The corresponding codeword $c(x) \in \cone{2}$ satisfies $\wt(c(x))=\bar{d}_{n_1,q}$. Now iii) is proved.

iv) It was noted in \cite{Ding1} that $\cone{2}$ is a duadic code. When $n_1 \equiv n_2 \equiv -1 \pmod{4}$, then $\mu_{-1}$ defines a splitting of $\cone{2}$, under the condition the equation (\ref{4:duad}) is always solvable for $\gamma \in \F$, hence (see (6.11) of page 226 in \cite{Huf}) the extended code $\bar{\mathcal C}_{\vepsi}^{(2)}$ is self-dual. This completes the proof of Theorem \ref{4:thm5}.
\end{proof}

\subsection{Ding's $3^{\mathrm rd}$ construction}
We write
\[x^{n_1n_2}-1=\prod_{i=0}^{n_1n_2-1}\left(x-\theta^i\right)=(x-1) F^{(3)}_{1,1}(x)F^{(3)}_{1,-1}(x)F_{2,1}(x)F_{2,-1}(x)F_{3,1}(x)F_{3,-1}(x),\]
where for any $\epsi \in \{\pm 1\}$,
\begin{eqnarray}
\label{4:con2}
F^{(3)}_{1,\epsi}(x)&=&\prod_{\substack{1 \le i \le n_1n_2-1\\
\legn{i}{n_2}=\epsi\\
\gcd(i,n_1)=1}} \left(x-\theta^i\right),
\end{eqnarray}
and the functions $F_{2,\epsi}(x)$ and $F_{3,\epsi}(x)$ are defined in (\ref{4:f2}) and (\ref{4:f3}) respectively.

\begin{defn}[Ding's $3^{\mathrm rd}$ constructions] Assume 4.1)--4.3). For any $\vec{\epsi}:=(\epsi_1,\epsi_2,\epsi_3) \in \{\pm 1\}^3$, let
 \begin{eqnarray*} \label{4:d3g} g^{(3)}_{\vec{\epsi}}(x):=F^{(3)}_{1,\epsi_1}(x)F_{2,\epsi_2}(x)F_{3,\epsi_3}(x).\end{eqnarray*}
Then define $C^{(3)}_{\vec{\epsi}}$ to be the cyclic code of length $n_1n_2$ over $\F$ with the generator polynomial $g^{(3)}_{\vec{\epsi}}(x)$, that is,
\begin{eqnarray} \label{4:defn3} C^{(3)}_{\vec{\epsi}}=\left(g^{(3)}_{\vec{\epsi}}(x)\right) \subset \F[x]/(x^{n_1n_2}-1). \end{eqnarray}
\end{defn}

By switching the roles $n_1 \leftrightarrow n_2$, it is easy to see that Ding's $3^{\mathrm rd}$ construction is the same as Ding's $2^{\mathrm nd}$ construction, hence the properties of the codes $\cone{3}$ follow directly from Theorem \ref{4:thm4}. For the sake of completeness, we state the results explicitly below.
\begin{thm}  \label{4:thm6}
\begin{itemize}
\item[i)] For any $\vepsi \in \{\pm 1\}^3$, the cyclic code $\cone{3}$ over $\F$ given in (\ref{4:defn3}) has length $n_1n_2$ and dimension $(n_1n_2+1)/2$, satisfying
\[d\left(\cone{3}\right) > d_{n_2,q} > \sqrt{n_2}. \]

\item[ii)] Write $\{\pm 1\}^3=A_1 \cup A_2$ where
\begin{eqnarray*} A_1&=&\{(1,1,1),(1,1,-1),(-1,-1,1),(-1,-1,-1)\}, \\
A_2&=&\{(1,-1,1),(1,-1,-1),(-1,1,1),(-1,1,-1)\}. \end{eqnarray*}
Then for $i=1$ or $2$, the codes $\cone{3}$ for $\vepsi \in A_i$ are all permutation equivalent.

\item[iii)] Taking $\vepsi_1=(1,1,1) \in A_1$, then
\[d\left(\C^{(3)}_{\vepsi_1}\right) = d_{n_2,q}+1. \]

\item[iv)] If $n_1 \equiv n_2 \equiv -1 \pmod{4}$, and $q$ is a prime power such that $q=2^m$ or $q \equiv 1 \pmod{4}$, let $\gamma \in \F$ such that
\[1+\gamma^2 n_1n_2=0. \]
Then for any $\vepsi \in \{\pm 1\}^3$, the code $\bar{\mathcal{C}}_{\vepsi}^{(3)}$ given by
\[\bar{\mathcal{C}}_{\vepsi}^{(3)}:=\left\{(c_0,\ldots,c_{n_1n_2-1},c_{\infty}): c_{\infty}=-\gamma \, c(1), \, \, c(x)=\sum_{i=0}^{n_1n_2-1} c_ix^i\in \cone{3} \right\}\]
has parameters $\left[n_1n_2+1,\frac{n_1n_2+1}{2}\right]$ and is self-dual.

\end{itemize}
\end{thm}

\section{Application to convolutional codes} \label{sec5}

The class of convolutional codes was invented by Elias \cite{E} in 1955 and has been widely in use for wireless, space, and broadcast communications since the 1970s. Compared with linear block codes, however, convolutional codes are not so well understood. In particular, there are only a few algebraic constructions of convolutional codes with good designed parameters.

Convolutional codes can be defined as subspaces over the rational functional field $\F(D)$ (see \cite{JZ,MC}) or over the field of Laurent series $\F(\!(\!D\!)\!)$ (\cite{P2}), or as submodules over the polynomial ring $\F[D]$ (\cite{GL2,S}), and the theories are all equivalent. Here we follow the presentation of \cite{GL2}. Interested readers may consult the textbooks \cite{JZ,MC,P2} for more information.

\begin{defn}
Let $\F[D]$ be the polynomial ring and $\F(D)$ the field of rational functions. For any $1 \le k \le n$, let $G(D) \in \F[D]^{k \times n}$ be a $k \times n$ polynomial matrix with rank$_{\F(D)}G(D)=k$. The rate $k/n$ convolutional code $\C$ generated by $G(D)$ is defined as the set
\[\C=\left\{u(D) G(D) \in \F(D)^n: u(D) \in \F(D)^k\right\}. \]
The matrix $G(D)$ is called a \emph{generator matrix} or an \emph{encoder} of $\C$.
\end{defn}


Let $G(D)=(g_{i,j}(D)) \in \F[D]^{k \times n}$ be an encoder of $\C$. The $i$-th row degree $v_i$ of $G(D)$ is $v_i=\max_j \deg g_{i,j}(D)$. If the encoders $G(D)$ and $G'(D)$ generate the same code $\C$, then we say $G(D)$ and $G'(D)$ are equivalent encoders. It is known that each convolutional code $\C$ can be generated by a \emph{minimal basic encoder} $G(D)$, that is, $G(D)$ has a polynomial right inverse, and the sum $\sum_{i=1}^k v_i$ of the row degrees of $G(D)$  attains the minimal value among all encoders of $\C$. Let $G(D) \in \F[D]^{k \times n}$ be a minimal basic encoder of $\C$, then $G(D)$ is automatically \emph{noncatastrophic}, that is, finite-weight codewords can only be produced from finite-weight messages, and the set $\{v_i: 1 \le i \le k\}$ of the row degrees of $G(D)$ is invariant among all minimal basic encoders of $\C$. Hence we can define the \emph{degree} of $\C$ as $\delta:=\sum_{i=1}^k v_i$ by using the minimal basic encoder $G(D)$ of $\C$. We call this $\C$ an $(n,k,\delta)_q$ convolutional code.

Each $v(D) \in \F(D)^n$ can be expanded uniquely as a Laurent series $v(D)=\sum_{j=-\infty}^{\infty}v_jD^j$ where $v_j \in \F^n$ for each $j$. Define the weight $\wt(v(D))$ as
\[\wt(v(D)):=\sum_{j=-\infty}^{\infty} \wt(v_j).\]
Here $\wt(v_j)$ denotes the Hamming weight of the vector $v_j \in \F^n$. The \emph{free distance} of the convolutional code $\C \subset F(D)^n$ is defined as
\[\D:=\min\left\{\wt(v(D)): v(D) \in \C, v(D) \ne 0\right\}.\]
It is known that (see \cite{S})
\[\D:=\min\left\{\wt(v(D)): v(D) \in \C \cap \F[D]^n, v(D) \ne 0\right\}.\]

For a given convolutional code $\C$, the four parameters $n,k,\delta,\D$ are of fundamental importance because $k/n$ is the rate of the code, $\delta$ and $\D$ determine respectively the decoding complexity and the error correcting capability of $\C$ with respect to some decoding algorithms such as the celebrated Viterbi algorithm \cite{VI}. For these reasons, for given rate $k/n$ and $q$, generally speaking, it is desirable to construct convolutional codes with relatively small degree $\delta$ and relatively large free distance $\D$. 

There are only a few algebraic constructions of convolutional codes. One particularly useful construction of convolutional codes is given in \cite[Theorem 3]{Jus}, which we describe below.

For any positive integer $n$ and any $c(x) \in \F[x]$, write
\[c(x)=\sum_{i=0}^{n-1} c_i(x^n) x^i, \]
where $c_i(x) \in \F[x]$ for each $i$. This induces a map $\phi_n$ given by
\[ \phi_n: \F[x] \to \F[x]^n, \quad c(x) \mapsto (c_i(x))_{i=0}^{n-1}.\]
For simplicity, we write $\phi_n(c(x)) \in \F[x]^n$ as an $n \times 1$ column vector. It is easy to see that for any positive integers $m,n$, this map $\phi_n$ defines a permutation equivalence from $\F[x]/(x^{mn}-1)$ to $\left(\F[x]/(x^m-1)\right)^n$.

For any $a,b$ in some finite extension of $\F$, we say that $a$ and $b$ are $n$-equivalent or belong to the same $n$-equivalence class if $a^n=b^n$.

\begin{thm}(\cite[Theorem 3]{Jus}) \label{5:just} Let $m,n$ be positive integers such that $\gcd(nm,q)=1$. Let $\C=(g(x)) \subset \F[x]/(x^{nm}-1)$ be a cyclic code of length $nm$ with the generator polynomial $g(x) \in \F[x]$. If $g(x)$ has at most $n-k$ roots in each $n$-equivalence class, then the matrix $G(D) \in \F[D]^{k \times n}$ given by
\[G(D)^T:=\left[\phi_n\left(g(D)\right),\phi_n\left(Dg(D)\right), \ldots, \phi_n\left(D^{k-1}g(D)\right)\right]\]
is a minimal basic encoder, and the code generated by $G(D)$ is a rate $k/n$ convolutional code with the free distance $\D$ satisfying $\D \ge d(\C)$. Here $G(D)^T$ is the transpose of the matrix $G(D)$, and $d(\C)$ is the minimal distance of the code $\C$.
\end{thm}

We remark that Theorem \ref{5:just} has been used in the construction of MDS-convolutional codes for any rate $k/n$ and any degree $\delta$ over a large field $\F$ (\cite[Theorem 3.3]{S}). In light of Theorem \ref{5:just}, it is readily seen that cyclic codes of composite length $mn$ and the estimate of the minimal distance are very useful in the construction of convolutional codes with prescribed rate and free distance. We also remark that generally speaking, it is not an easy task to construct a generator polynomial $g(x)$ such that $g(x)$ has at most $n-k$ roots in each $n$-equivalence class and the corresponding cyclic code has large minimal distance at the same time (see \cite[Theorem 6]{Jus}). Using Theorem \ref{5:just}, we note that the  cyclic codes from (\ref{3:defn1}) and from Ding's $1^{\mathrm st}$ construction in (\ref{4:defn1}) are useful to construct convolutional codes with large free distance.

\begin{thm} \label{5:thm9} Let $\F$ be a finite field of order $q$.
\begin{itemize}
\item[(i)] Let $n$ be a prime number such that $\legn{q}{n}=1$, and $r \ge 2$ be a positive integer such that $\gcd(nr,q)=\gcd(n,r)=1$. Let $\vepsi \in \Lambda \setminus
\{ (0,\ldots,0), (1,\ldots,1)\} \subset \{\pm 1\}^r$ where $\Lambda$ is given in (\ref{3:Lambda}), and let $\C_{\vepsi}=(g_{\vepsi}(x)) \subset \F[x]/(x^{nr}-1)$ be the cyclic code of length $nr$ given in (\ref{3:defn1}). Then for any $1 \le k \le (n+1)/2$, the code generated by the encoder $G(D)$ given by
\[G(D)^T:=\left[\phi_n\left(g_{\vepsi}(D)\right),\phi_n\left(Dg_{\vepsi}(D)\right), \ldots, \phi_n\left(D^{k-1}g_{\vepsi}(D)\right)\right]\]
is a rate $k/n$ convolutional code with the free distance satisfying $\D >\sqrt{n}$.

\item[(ii)] Let $n_1,n_2$ be two distinct odd primes such that $\legn{q}{n_1}=\legn{q}{n_2}=1$ and $\gcd(n_1n_2,q)=1$. For any $\vepsi \in
\{\pm 1\}^3$, let $\cone{1}=(\gone{1}(x)) \subset \F[x]/(x^{n_1n_2}-1)$ be the cyclic code of length $n_1n_2$ given in (\ref{4:defn1}). Then for any $1 \le k \le (n_1-1)/2$, the code generated by the encoder $G(D)$ given by
\[G(D)^T:=\left[\phi_{n_1}\left(g_{\vepsi}(D)\right),\phi_{n_1}\left(Dg_{\vepsi}(D)\right), \ldots, \phi_{n_1}\left(D^{k-1}g_{\vepsi}(D)\right)\right]\]
is a rate $k/n_1$, degree $\frac{(n_2-1)k}{2}$ convolutional code over $\F$ with the free distance satisfying $\D >\max\left\{\sqrt{n_1},\sqrt{n_2}\right\}$.
\end{itemize}
\end{thm}
\begin{proof} The only thing we need to check is that the generator polynomial $g(x)$ in (i) and (ii) has at most $n-k$ roots in each $n$-equivalence class, which is almost trivial, so we omit the details.
\end{proof}

\begin{exam} For $q=2,n=7,r=3$, using the cyclic code in (\ref{3:defn1}), for $\vepsi=(1,-1,-1)$, we find
\begin{eqnarray*} g_{\vepsi}(x)&=& x^{9} + x^{8} + x^{7} + x^{5} + x^{4} + x +    1.\end{eqnarray*}
The code $\C_{\epsi}=(g_{\vepsi}(x)) \subset \mathbb{F}_2[x]/(x^{21}-1)$ has parameters $[21,12,5]$. We may write
\[g_{\vepsi}(x)= \sum_{i=0}^6 g_i(x^7) x^i,\]
where
\[\begin{array}{l}
g_0(x)=1+x,\,\, g_1(x)=1+x,\,\,
g_2(x)=x,\\
g_3(x)=0,\,\,
g_4(x)=1,\,\,
g_5(x)=1,\,\,
g_6(x)=0.
\end{array}\]
Taking $k=4$, from (1) of Theorem \ref{5:thm9}, the matrix $G(D) \in \mathbb{F}_2[D]^{3 \times 7}$ given by
\[G(D)=\left[\begin{array}{ccccccc}
1+D&1+D&D&0&1&1&0\\
0&1+D&1+D&D&0&1&1\\
D&0&1+D&1+D&D&0&1\\
D&D&0&1+D&1+D&D&0
\end{array}\right]\]
is a minimal basic encoder which generates a rate 4/7, degree 4 convolutional code $\C$ over $\mathbb{F}_2$ with free distance $\D \ge 5$. Actually the codeword $c=[1,1,1,0]*G(D)=[1,0,D,1,1+D,0,0]$ has weight $5$, so $\D=5$. We would like to mention that the Heller bound \cite[Corollary 3.18]{JZ} implies that the largest free distance $\D$ of any binary, rate $4/7$ and unit-memory convolutional code satisfies $\D \le 7$.

\end{exam}
\begin{exam} For $q=2,n_1=7,n_2=17$, using Ding's $1^{\mathrm st}$ construction (see (\ref{4:defn1})), for $\vepsi=(1,1,1)$, we find
\begin{eqnarray*} \gone{1}(x)&=& x^{59} + x^{58} + x^{57} + x^{56} + x^{51} + x^{50}+ x^{49} + x^{48} + x^{47} + x^{46} + x^{45}
    + x^{44} \\
&&    + x^{43} + x^{39}
    + x^{38} + x^{37} + x^{36} + x^{33}
+ x^{32} + x^{30} + x^{24} + x^{22}    \\
&&+ x^{20} + x^{19}
+ x^{18} + x^{15} + x^{14} + x^{9} + x^{8} + x^{7} + x^{6} + x^{5} + x^{4} + x +    1.\end{eqnarray*}
The code $\cone{1}=(\gone{1}(x)) \subset \mathbb{F}_2[x]/(x^{119}-1)$ has parameters $[119,60,12]$. We may write
\[\gone{1}(x)= \sum_{i=0}^6 g_i(x^7) x^i,\]
where
\[\begin{array}{l}
g_0(x)=1+x+x^2+x^7+x^8,\\
g_1(x)=1+x+x^2+x^3+x^5+x^6+x^7+x^8,\\
g_2(x)=x+x^4+x^5+x^6+x^7+x^8,\\
g_3(x)=x^3+x^5+x^6+x^8,\\
g_4(x)=1+x^2+x^4+x^5+x^6,\\
g_5(x)=1+x^2+x^4+x^6,\\
g_6(x)=1+x^2+x^6.
\end{array}\]
Taking $k=3$, from (2) of Theorem \ref{5:thm9}, the matrix $G(D) \in \mathbb{F}_2[D]^{3 \times 7}$ given by
\[G(D)=\left[\begin{array}{lllllll}
g_0(D)&g_1(D)&g_2(D)&g_3(D)&g_4(D)&g_5(D)&g_6(D)\\
Dg_6(D)&g_0(D)&g_1(D)&g_2(D)&g_3(D)&g_4(D)&g_5(D)\\
Dg_5(D)&Dg_6(D)&g_0(D)&g_1(D)&g_2(D)&g_3(D)&g_4(D)
\end{array}\right]\]
is a minimal basic encoder which generates a rate 3/7, degree 24 convolutional code $\C$ over $\mathbb{F}_2$ with free distance $\D \ge 12$.

\end{exam}
\section{Conclusion}\label{conclude}

In this paper, we develop a general method to study cyclic codes of composite length and to estimate the minimal distance (Theorem \ref{2:thm1}), which can be used to study the cyclic codes from Ding's three constructions (Theorems \ref{4:thm4}, \ref{4:thm5}, \ref{4:thm6}). We also give a general construction of cyclic codes of composite length which are related to quadratic residue codes of prime length and study some of their properties (Theorems \ref{3:thm2} and \ref{3:thm3}). The cyclic codes in this paper can be used to construct convolutional codes with large free distance. While we have obtained an ``almost'' square-root bound on the minimal distance of all these cyclic codes that we are interested in, numerical data shows that the minimal distance of these codes should be much larger. So we list the following as open questions:

\begin{itemize}
\item[(1)] For the code $\C_{\vepsi}$ given in (\ref{3:defn1}), is it possible to improve the lower bound on $d \left(\C_{\vepsi}\right)$ in Theorem \ref{3:thm3} when $\vepsi \ne (0,\ldots,0)$ or $(1,\ldots,1)$?

\item[(2)] For the codes $\cone{i}$ ($i=1,2,3$) from Ding's three constructions, is it possible to improve the lower bound on $d \left(\cone{i}\right)$ from Theorems \ref{4:thm4}, \ref{4:thm5} and \ref{4:thm6}? In particular is it true that $d \left(\cone{1}\right) \ge \sqrt{n_1n_2}$ for any $\vepsi \in
\{\pm 1\}^3$ and $d \left(\cone{i}\right) \ge \sqrt{n_1n_2}$ for any $i=2,3$ and $\vepsi =(1,-1,-1)$?

\end{itemize}






\subsection*{Acknowledgments} The author is grateful to Professor Cunsheng Ding for many useful suggestions of the paper. In particular the numerical computation by the end of Section \ref{sec3} can not be done so efficiently without Professor Ding's generous sharing of his original Magma codes from the work \cite{Ding1}. 


\end{document}